\documentclass[a4paper]{article}
\usepackage{fullpage}

\usepackage[T1]{fontenc}
\usepackage{amsthm}
\usepackage{amsmath}
\usepackage{amssymb}
\usepackage{amsfonts}
\usepackage{microtype}
\usepackage{cite}
\usepackage[caption=false]{subfig}
\usepackage{xspace}
\usepackage{hyperref}
\usepackage{color}
\usepackage{etoolbox}
\usepackage{lmodern}
\usepackage{graphicx}
\graphicspath{{./graphics/}}

\theoremstyle{plain}
\newtheorem{theorem}{Theorem}[section]
\newtheorem{lemma}[theorem]{Lemma}
\newtheorem{observation}[theorem]{Observation}

\newtheorem{construction}[theorem]{Construction}
\newtheorem{definition}[theorem]{Definition}

\newcommand{\ER}[0]{$\exists\mathbb{R}$\xspace}
\newcommand{\SECTOR}[0]{\textsc{Sector}\xspace}
\newcommand{\STRETCH}[0]{\textsc{Simple-Stretchability}\xspace}
\newcommand{\NSTRETCH}[0]{\textsc{Stretchability}\xspace}

\newcommand{\R}[0]{\mathbb{R}}

\newcommand{\eps}{\varepsilon}

\newcommand{\Wlog}{W\@.l\@.o\@.g\@.\xspace}

\DeclareMathOperator{\dist}{dist}

\usepackage{cleveref}
\crefname{construction}{Construction}{Constructions}
\Crefname{construction}{Construction}{Constructions}
\Crefname{observation}{Observation}{Observations}
\title{Recognizing Generalized Transmission Graphs of Line 
  Segments and Circular Sectors\thanks{Supported in part 
  by ERC StG 757609.}}

\author{%
  Katharina Klost\footnote{%
    Institut f\"ur Informatik, Freie Universit\"at Berlin, 
    Germany \texttt{\{kathklost, mulzer\}@inf.fu-berlin.de}}
   \and
Wolfgang Mulzer\footnotemark[2]}

\begin{document}

\maketitle
\begin{abstract}
Suppose we have an arrangement $\mathcal{A}$ of $n$ geometric 
objects $x_1, \dots, x_n \subseteq \R^2$ in the plane, with a
distinguished point $p_i$ in each object $x_i$.
The \emph{generalized transmission graph} of 
$\mathcal{A}$ has vertex set $\{x_1, \dots, x_n\}$ and  a 
\emph{directed} edge $x_ix_j$ if and only if $p_j \in x_i$.
Generalized transmission graphs provide a generalized model of
the connectivity in networks of directional antennas.

The complexity class $\exists \mathbb{R}$ contains all problems that 
can be reduced in polynomial time to an existential sentence of the 
form $\exists x_1, \dots, x_n: \phi(x_1,\dots, x_n)$, 
where $x_1,\dots, x_n$ range over $\mathbb{R}$ and $\phi$ is 
a propositional formula with signature $(+, -, \cdot, 0, 1)$. 
The class $\exists \mathbb{R}$ aims to capture the complexity of 
the \emph{existential theory of the reals}. It 
lies between $\mathbf{NP}$ and $\mathbf{PSPACE}$.

Many geometric decision problems, such as recognition of 
disk graphs and of intersection graphs of lines, 
are complete for $\exists \mathbb{R}$. 
Continuing this line of research, we show that the recognition 
problem of generalized transmission graphs of line segments and 
of circular sectors is hard for $\exists \mathbb{R}$.
As far as we know, this constitutes the first such result for a 
class of \emph{directed} graphs.
\end{abstract}

\section{Introduction}

Let $\mathcal{A}$ be an arrangement of $n$ geometric objects
$x_1, \dots, x_n$ in the plane. The \emph{intersection graph} of 
$\mathcal{A}$ has one vertex for each object and an undirected edge 
between two objects $x_i$ and $x_j$ if and only if $x_i$ and $x_j$ 
intersect. In particular, if the objects are (unit) disks, we
speak of \emph{(unit) disk graphs}. These are often used as a 
symmetric model for antenna reachability. In some cases, however, 
this symmetry is not desired, since it does not accurately model the
properties of the network. For omnidirectional antennas, there is 
an asymmetric model called 
\emph{transmission graphs}~\cite{kaplan_spanners_2015}.
Transmission graphs are also defined on disks: 
as in disk graphs, there is one vertex per disk, 
and the edges indicate directed reachability. There is a \emph{directed} 
edge between two disks if the first disk contains the center of the second disk. 

Here, we present a new class of \emph{generalized transmission graphs}. 
Now, the objects may be arbitrary sets in $\R^2$, and the 
points that decide about the existence of an edge can be 
arbitrary points in the objects. 

For a given graph class, the \emph{recognition problem} is as follows:
given a combinatorial graph $G = (V, E)$, decide whether $G$ belongs
to this class. For the recognition of geometrically defined graphs, 
it turned out that the complexity class \ER plays a major role.
The class \ER was formally introduced by 
Schaefer~\cite{schaefer_complexity_2009}. It consists of all problems 
that are polynomial-time reducible to the set of all true sentences 
of the form $\exists x_1,\dots, x_n: \Phi(x_1,\dots,x_n)$. Here, 
$\Phi$ is a quantifier-free formula with signature 
$(+,-,\cdot, 0,1)$ additional to the standard boolean signature. The variables range over the reals. Hardness for this class is defined via polynomial reduction.

There are multiple classes of intersection graphs for which
the recognition problem is \ER-complete. Kang and M\"uller 
showed this for intersection graphs of 
$k$-spheres~\cite{kang_sphere_2012}, and Schaefer 
proved a similar result for intersection graphs of 
line segments and convex sets~\cite{schaefer_complexity_2009}.

One prototypical \ER-complete problem that serves
as the starting point of many reductions is \NSTRETCH, 
which was among the first known \ER-hard problems. 
The original hardness-proof is due to 
Mn\"ev~\cite{mnev_realizability_1985}, and it was restated 
in terms of \ER by Matou\v{s}ek~\cite{matousek_intersection_2014}. 

Here, we show that the recognition of generalized 
transmission graphs of line segments and of a certain
type of arrangements of circular sectors is hard for \ER.
For this, we need to extend the known proofs significantly, and we need to develop new tools to reason about geometric realizations of directed graphs.
With some further work the inclusion of these problems in \ER could be shown. For details see the master thesis of the first author \cite{klost_complexity_2017}.

\section{Preliminaries}

\subsection{Graph classes}
Let $x_1, \dots, x_n \subseteq \mathbb{R}^2$ be a set of $n$
objects, and suppose that there is a distinguished point
$p(x_i) \in x_i$, in every object $x_i$. The 
\emph{generalized transmission graph} of these objects 
is a directed graph $G = (V, E)$ with 
\[
V =\{x_1, \dots, x_n\}
\text{ and }E=\{(x_i,x_j) \mid p(x_j) \in x_i, 1 \leq i, j \leq n\}.
\]

We will consider generalized transmission graphs for 
line segments and circular sectors. In these cases,
the distinguished points $p(x_i)$ are defined as follows: 
for line segments, we choose one fixed endpoint; 
for circular sectors, we choose the apex. 

When constructing arrangements of line segments and of circular 
sectors below, in \Cref{sec:linesegments,sec:circularsectors}, 
we need some notation.
A line segment $\ell$ is described by an \emph{endpoint} 
$p(\ell)$, a \emph{length} $r(\ell)$, and a \emph{direction} $u(\ell)$.
A \emph{circular sector} $c$ is presented by an \emph{apex} 
$p(c)$, a \emph{radius} $r(c)$, an \emph{opening angle} 
$\alpha(c)$, and a direction $u(c)$. 
The direction is a vector in $\R^2$, and it indicates 
the direction of the bisector. We will call the bounding 
line segments the \emph{outer} line segments of $c$. Let 
$B(c)$ be the smallest rectangle with two sides parallel 
to $u(c)$ that contains $c$, the \emph{bounding box} of $c$.

\subsection{Stretchability and combinatorial descriptions}

Let $\mathcal{L}$ be an arrangement of $n$ non-vertical lines,
such that no two lines in $\mathcal{L}$ are parallel. We define 
the \emph{combinatorial description} $D(\mathcal{L})$
of $\mathcal{L}$ as follows: 

Let $g$ be a vertical line that lies to the left of all 
intersection points of $\mathcal{L}$. We number the lines 
$\ell_1, \dots, \ell_n$ in the order in which they intersect $g$, 
from top to bottom. This ordering corresponds to the ascending order 
of the slopes. For each line $\ell_i$, $i = 1, \dots, n$, 
we have a list $O^i$ of the following form:
\begin{align*}
O^i& = (o^i_1, \dots, ,o^i_k) & o_j^i &\subseteq \{1, \dots, ,n\}\\
\bigcup_{j = 1}^k o_j^i &= \{1,...,n\}& o_j^i \cap o_{j'}^{i} &= 
\emptyset, \text{ for } j \neq j'.
\end{align*}
For $i = 1, \dots, n$, the order of the indices in 
$O^i$ indicates the order in which the lines $\ell_j$ cross 
$\ell_i$, as we travel along $\ell_i$ from left to right. 
The lists $O^i$, for $i = 1, \dots, n$, form the 
\emph{combinatorial description}  of the arrangement $\mathcal{L}$. 
If $\mathcal{L}$ is simple, each $o_j^i$ is a singleton. 

Given a combinatorial description $\mathcal{D}$ as above,
it is relatively easy to detect whether it comes from an
arrangement of \emph{pseudo-lines}. This can be done by checking
a few simple axioms~\cite{Knuth92}.
However, the decision problem \NSTRETCH
of deciding if $\mathcal{D}$ originates from an
actual arrangement of \emph{lines} turns out to be
significantly harder.
If all sets $o_j^i$ are singletons, the same problem is called \STRETCH. 
Both variants of the problem are complete for 
\ER~\cite{matousek_intersection_2014,mnev_realizability_1985}.

\section{Line segments}\label{sec:linesegments}

We now present our first result on the recognition of
intersection graphs of line segments.

\begin{theorem}\label{thm:1darc}
Recognizing a generalized transmission graph
of line segments is \ER-hard.
\end{theorem}

\begin{proof}
The proof proceeds by a reduction from \STRETCH.
Given an alleged description $\mathcal{D}$ 
of a simple arrangement of lines, we construct a graph 
$G_L = (V_L, E_L)$ such that $\mathcal{D}$ 
is realizable as a line arrangement if and only if $G_L$ 
is the generalized transmission graph of an arrangement 
of line segments. We set $V_L = A \cup B \cup C$ with 
\begin{alignat*}{2}
A &= \{a_{\{i, k\}} && \mid 1 \leq i \neq k \leq n \}, \\
B &= \{b^i_{k} && \mid 1 \leq i \leq n, 1 \leq k \leq n-1\},\\
C &= \{c_i && \mid 1 \leq i \leq n\},
\end{alignat*}
where the $c_i$ are numbered in order given by 
$\mathcal{D}$.  The $\{\ \}$ in the indices 
of the $a_{\{i,k\}}$ indicates that $a_{\{i, k\}} = a_{\{k,i\}}$.

Before defining the edges, we describe the 
intuitive meaning of the different vertices. 
The line segments associated with $C$  correspond 
to the lines $\ell_i$ of the arrangement. The 
endpoints of the line segment associated with 
$a_{\{i,k\}}$ will enforce that there is an intersection 
of the line segments for $c_i$ and $c_k$, for 
$1 \leq i \neq k \leq n$. The endpoints of the 
line segments for the $b_{k}^i$, $k = 1, \dots, n - 1$, 
will be placed between the $a_{\{i, k\}}$ on $c_i$ and 
thus enforce the order of the intersection. When it 
is clear from the context, we will not explicitly 
distinguish between a vertex of the graph and the 
associated line segment. Now we define the edges: 
\begin{alignat*}{2}
E_L &= \phantom{\cup} \{(c_i, a_{\{i,k\}}), (c_i, b^i_k), (b^i_k,c_i)
&&\mid 1 \leq i \neq k \leq n\}\\
&\phantom{=} \cup \{(b^i_{o^i_k}, b^i_{o^i_{l}}), 
(b^i_{o^i_k}, a_{\{i,o^i_l\}})&&\mid 
1 \leq i \leq n, 1 \leq l < k\leq n-1 \}
\end{alignat*}
Given $\mathcal{D}$, the sets $V_L$ and $E_L$ 
can be constructed in polynomial time. 
It remains to show correctness. 
Suppose first that
$\mathcal{D}$ is realizable, and 
let $\mathcal{L} = (\ell_1, \dots, \ell_n)$ 
be a simple line arrangement with
$\mathcal{D} = \mathcal{D}(\mathcal{L})$. We show that 
there exists an arrangement $\mathcal{C}$ of line 
segments that realizes $G_L$.
Let $D$ be a disk that contains all vertices of 
$\mathcal{L}$, with $\partial D$ having a positive 
distance from each vertex.

The circular order of the intersections between 
$\ell_1, \dots, \ell_n$ and $\partial D$ is \allowbreak 
$\ell_1, \dots, \ell_n, \allowbreak \ell_1, \dots, \ell_n$. 
There is no vertical line in $\mathcal{L}$, so we can add 
a virtual vertical line $\ell'$ that divides the intersection 
points along $\partial D$ into a ``left'' set 
$D_l = \{q_1^l, q_2^l, \dots, q_n^l\}$ and a ``right'' set
$D_r =\{ q_1^r, q_2^r, \dots, q_n^r\}$ such that each 
set contains exactly one intersection with each line
$\ell_i$, $i = 1, \dots, n$.

For $i = 1, \dots, n$, we set $c_i$ to
$\ell_i \cap D$, with $p(c_i) = q_i^l$.
The $a_{\{i, k\}}$ are constructed such that 
$p(a_{\{i, k\}})$ is the intersection point of $\ell_i$ 
and $\ell_k$. The direction and length are chosen 
in such a way that $a_{\{i, k\}}$ intersects no other lines.
Now we place the line segments $b_{o_k^i}^i$. 
They are positioned such that $p(b_{o_k^i}^i)$ 
lies between $p(a_{\{i, o_k^i\}})$ and $p(a_{\{i, o_{k+1}^i\}})$, 
for $k = 1, \dots, n-2$. Furthermore, we place $p(b_{o_{n-1}})$ 
to the right of $a_{\{i, o_{n-1}^i\}}$. The line 
segments lie on the lines $\ell_i$ such that $p(c_i)$ lies in the
relative interior of $b_{k}^i$. For an example of this 
construction, see \Cref{fig:segmentconstr}.
 \begin{figure}

\subfloat[Complete line segment construction for three lines]{
	\includegraphics[width=0.45\textwidth]{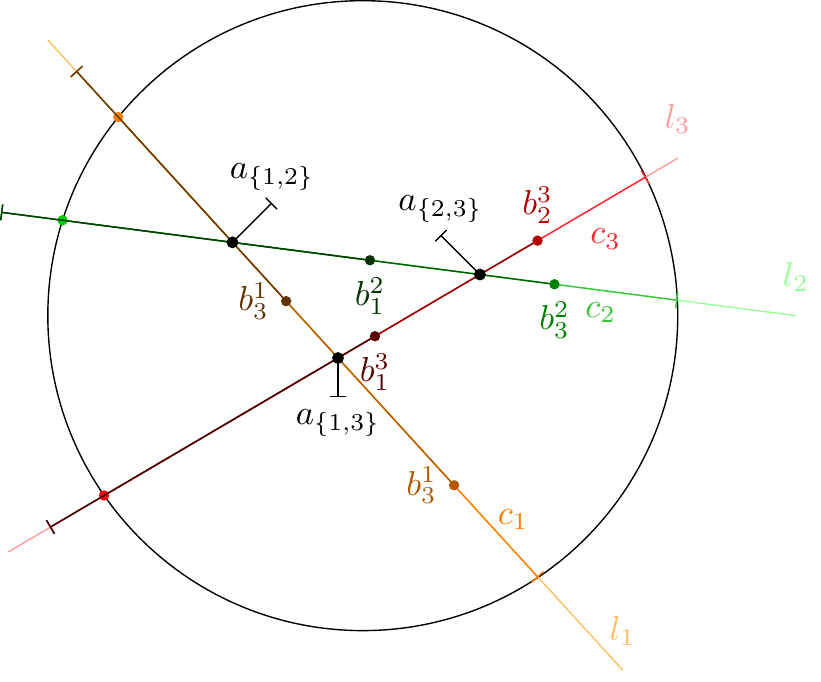}
}\qquad
\subfloat[Closeup of $c_2$. The line segments $b_1^2$ and $b_3^2$ 
are shifted upwards to show their positioning.]{

\includegraphics[width=0.45\textwidth]{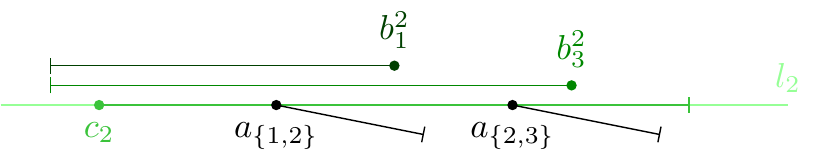}
}
\caption{Construction of the line segments.}
\label{fig:segmentconstr}
\end{figure}
It follows from the construction that the generalized 
transmission graph of $\mathcal{C}$ is indeed $G_L$. 

Now consider an arrangement $\mathcal{C}$ of line 
segments realizing $G_L$. 
Let $\mathcal{L}' =(\ell_1', \dots ,\ell_n')$ 
be the arrangement of lines where $\ell_i'$ is
the supporting line of $c_i$, for $i = 1, \dots, n$.
We claim that $\mathcal{D} = \mathcal{D}(\mathcal{L}')$.

We first consider the role of the line segments $a_{\{i, k\}}$. 
Since $p(a_{\{i, k\}})$ lies on $c_i$ and $c_k$, we have 
$p(a_{\{i, k\}}) = c_i \cap c_k$, and therefore $\ell'_i$ 
and $\ell_k'$ intersect in $p(a_{\{i, k\}})$. This 
ensures that all pairs of lines have an intersection point 
that is also the endpoint of an $a_{\{i,k\}}$.
Next, we have to show that the order of the intersections 
along each line $\ell'_i$, for $i = 1, \dots, n$, is in the 
order as given by $\mathcal{D}$. This is guaranteed by the 
line segments $b^i_k$ as follows:
By the definition of $E_L$, namely by the edges 
$(c_i, b_k^i)$ and $(b_k^i, c_i)$, it is ensured that 
all $p(b_k^i)$ lie on the same line as $c_i$. The definition 
also enforces the order of the $p(a_{\{i, k\}})$ and $p(b_k^i)$ 
along the line. Since $p(a_{\{i, o_k\}})$ lies on $b_{o_{k + 1}}^i$
but not on $b_{o_k}^i$ and since all lie on the same line $c_i$, 
it has to lie between the corresponding endpoints. 
This enforces the correct order of the intersections.
\end{proof}

\section{Circular sectors}\label{sec:circularsectors}

We now consider the problem of recognizing generalized 
transmission graphs of circular sectors. The reduction extends 
the proof for \Cref{thm:1darc}, but we need to be more
careful in order to enforce the correct order of intersection. 

We will only consider circular sectors with opening
angle $\alpha \leq \pi/4$.
If $x$ and $y$ are circular sectors with $p(x) \in y$ and 
$p(y) \in x$, we call $x$ and $y$ a \emph{mutual couple} 
of circular sectors. We write $\gamma(u(x), u(y))$ for 
the counter-clockwise angle between the vectors $u(x)$ and 
$u(y)$.

\begin{observation}\label{ob:pairsangles}
Let $x$ and $y$ be a mutual couple of circular sectors, then
\[
  |\pi - \gamma(u(x), u(y))| \leq (\alpha(x)+ \alpha(y))/2.
\]
\end{observation}

The argument is visualized in \Cref{fig:coupleangle}.

 \begin{figure}
\subfloat[Extreme position of $x$ and $y$; 
 the symmetric case is indicated by the red line.]{
 \includegraphics[width=0.45\textwidth]{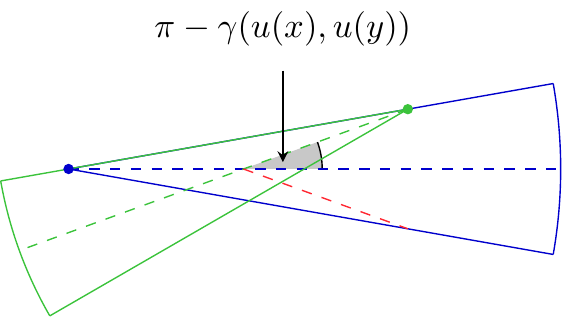}
 \label{fig:coupleangle}
}
\qquad
  \subfloat[$a_k$ and $l_i$ form a mutual couple, so $u(a_k)$ 
 lies in the blue range.  The apex of $a_{k-1}$ is projected to 
 the right of $p(a_k)$, forcing $u(a_k)$ to be in the red range.]{
 \includegraphics[width=0.45\textwidth]{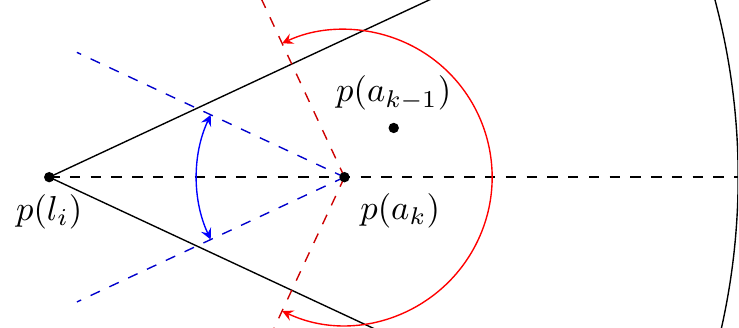}
 \label{fig:ordergadget}
}
 \end{figure}

\begin{observation}\label{ob:outerangle}
Let $x$ and $y$ be circular sectors whose bisectors intersect 
at an acute angle of $\beta > \max\{\alpha(x), \alpha(y)\}/2$.
Then, the acute angle between the outer line segments 
of $x$ and the bisector of $y$ is at least 
$\beta - \max\left\{\alpha(x), \alpha(y)\right\}/2$.
\end{observation}

\begin{lemma}\label[lemma]{la:orderingGadget}
Let $l$ be a circular sector and let $a_1,\dots, a_n$ 
be circular sectors with
\begin{alignat*}{2}
p(a_i) &\in l, &~&1\leq i \leq n,\\
p(a_i) &\in a_j,&~&1\leq i <j \leq n, \text{ and}\\
p(l) &\in a_j, &~&1\leq j \leq n.
\end{alignat*}
Then, the projection of the $p(a_i)$ onto the directed 
line $\ell$ defined by $u(l)$ has the order 
\[
 O=o_1, \dots, o_{n}  = a_1, \dots, a_n.
\]
\end{lemma}

\begin{proof}
Each $a_i$ forms a mutual couple with $l$. Thus, with 
\Cref{ob:pairsangles}, we get
\begin{align}
|\pi - \gamma\left( u(a_i), u(l) \right)| &\leq \pi/4\label{eq:ajsmall}.
\end{align}
Assume that the order of the projection differs 
from $O$. Let $O' = o'_1,\dots, o'_n$ be the actual order 
of the projection of the $p(a_i)$ onto $\ell$. Let $j$ 
be the first index with $o_j' \neq o_j$ and $o'_j = a_k$. Then, 
there is an $o'_i$, $i >j$, with $o'_i = a_{k-1}$. By 
definition, $p(a_{k-1})$ has to be included in $a_{k}$, 
while still being projected on $\ell$ to the right of 
$p_{k}$. This is only possible if 
\[
|\pi -\gamma(u(a_j), \ell)| > \frac{\pi}{2} -\frac{\alpha(a_k)}{2} 
\geq \frac{\pi}{2} - \frac{\pi}{8} =\frac{3\pi}{8} > \frac{\pi}{4}
\]
This is a contradiction to (\ref{eq:ajsmall}), and consequently 
the order of the projection is as claimed.
The possible ranges of the angles are illustrated in 
\Cref{fig:ordergadget}.
\end{proof}

An arrangement $\mathcal{C}$ of circular sectors is called 
\emph{equiangular} if $\alpha(c) = \alpha(c')$ for all 
circular sectors $c, c' \in \mathcal{C}$.

Let $c, c'$ be two circular sectors of $\mathcal{C}$, 
and assume that $d \in \mathcal{C}$ is a circular sector
with $p(d) \in c$ and $p(d) \in c'$, such that $c$ and 
$c'$ do not form both a mutual couple with the same circular sector. 
Moreover let $\beta_\text{min}$ be the smallest acute angle 
between the bisector of any pair $c, c'$ with this property.
We will call the arrangement \emph{wide spread} if 
\[
\beta_\text{min} \geq 2\cdot \max_{c \in \mathcal{C}}(\alpha(c))
\]
The possible situations are depicted in \Cref{fig:widespread}.
\begin{figure}
\subfloat[Constraint on the angle.]{
 \includegraphics[width=0.45\textwidth]{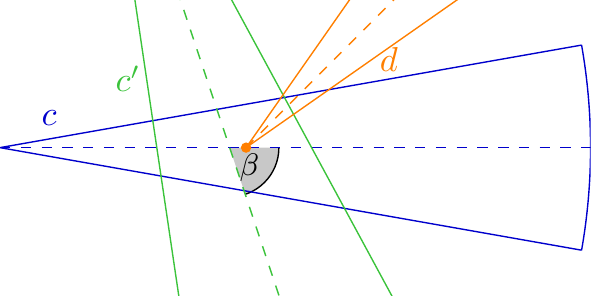}
}\qquad
  \subfloat[No constraint on the angle.]{
 \includegraphics[width=0.45\textwidth]{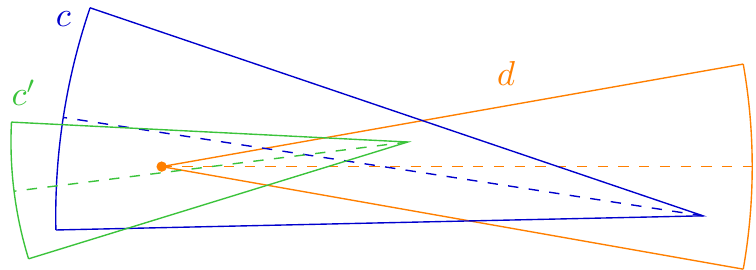}

}
\caption{The wide spread condition.}
\label{fig:widespread}
 \end{figure}

\begin{definition}
   The recognition problem of the generalized transmission graphs 
   of equiangular, wide spread circular sectors  is called \SECTOR.
\end{definition}

Now we want to show that \SECTOR is hard for \ER. This is 
done in three steps. First, we give a polynomial-time 
construction that creates an arrangement of circular sectors 
from an alleged combinatorial description of a line arrangement. 
Then we show that this construction is indeed a reduction 
and therefore show the \ER-hardness of \SECTOR.

\begin{construction}\label{co:cs}
Given a description $\mathcal{D}$ where all $o_i$ are singletons, 
we construct a graph $G_L = (V_L, E_L)$.
For this construction, let $1 \leq i, k, l\leq n$,
$1 \leq m, m', m'' \leq 3$. The set of vertices is defined as follows:
\begin{alignat*}{2}
V_L &=\phantom{\cup}\{c_{im}\} \cup \{a^{im}_{km'}&\mid i\neq k\} 
\cup \{b^{im}_{km'}&\mid  i\neq k\}
\end{alignat*}
As for the line segments, we do not distinguish 
between the vertices and the circular sectors. 
For the vertices $a_{km'}^{im}$ and $b_{km'}^{im}$, 
the upper index indicates the $c_{im}$ with whom 
$a_{km'}^{im}$ and $b_{km'}^{im}$ form a mutual couple. 
The lower index hints at a relation to $c_{km'}$.
In most cases, the upper index is $im$ and the lower index 
differs. For better readability, the indices are 
marked bold ($a_{\mathbf{im}}^{\mathbf{km'}}$), 
if $im$ is the lower index.

The bisectors of the circular sectors $c_{i2}$ 
will later define the lines of the arrangement. 
The circular sectors $a^{im}_{km'}$ and 
$a_{\mathbf{im}}^{\mathbf{km'}}$ have a similar role as 
the $a_{\{i,k\}}$ in the construction for the line segments. 
They enforce the intersection of $c_{im}$ and $c_{km'}$. 
Similar to the $b_{k}^{i}$, the $b_{km'}^{im}$ 
help enforcing the intersection order.

We describe $E_L$ on a high level. For a detailed technical 
description, refer to \Cref{ap:edgedef}.
We divide the edges of the graph into \emph{categories}. The 
first category, $E_I$, contains the edges that enforce an 
intersection between two circular sectors $c_{im}$ and 
$c_{km'}$, for $k <l$. The edges of the next category $E_C$ 
enforce that each $a_{km'}^{im}$ and each $b_{km'}^{im}$ 
forms a mutual couple with $c_{im}$.
\begin{align*}
\begin{alignedat}{2}
E_I=\phantom{\cup}&\{(c_{im}, a_{km'}^{im})&&\mid i\neq k\}\\
	\phantom{=}\cup&\{(c_{im}, a^{km'}_{im}) &&\mid i\neq k\}
	\end{alignedat}
	&&
	\begin{alignedat}{2}
	E_C=\phantom{\cup}&\{(a_{km'}^{im}, c_{im})&&\mid i\neq k\}\\
	\phantom{=}\cup&\{(c_{im}, b_{km'}^{im})&&\mid i \neq 
	k\}\\
	\phantom{=}\cup&\{(b_{km'}^{im}, c_{im})&&\mid i \neq k\}
	\end{alignedat}
\end{align*}
The edges in the next categories enforce the local order. 
The first category, called $E_\text{GO}$, enforces a global 
order in the sense that the apexes of all 
$a_{{o_j}m'}^{im}$ and $b_{{o_j}m'}^{im}$ 
will be projected to the left of any $a_{{o_k}m'}^{im}$ and 
$b_{{o_k}m'}^{im}$   with $k >j$. Additionally, all 
$a_{\mathbf{im}}^{\mathbf{o_jm'}}$ will be included in 
$a_{{o_k}m'}^{im}$ and $b_{{o_k}m'}^{im}$. The projection order 
is enforced by the construction described in \Cref{la:orderingGadget}, 
the inclusion is enforced by adding the appropriate edges. 

It remains to consider the local order of the six 
circular sectors ($a_{j1}^{im}, \dots, a_{j3}^{im}, \allowbreak 
b_{j1}^{im}, \dots, b_{j3}^{im}$) that are associated 
with $c_{im}$ for each intersecting circular sector 
$c_{j2}$. The projection order of these is either ``$1$, $2$, $3$'' 
or ``$3$, $2$, $1$'', depending on the order of 
$l_i$ and $l_j$ on the vertical line. If $l_j$ is below $l_i$, 
the order on $c_{im}$ is ``$1$, $2$, $3$''; in the other case, 
it is ``$3$, $2$, $1$''.  This is again enforced by adding the 
edges as defined in \Cref{la:orderingGadget}.
For a possible realization of this graph, 
see \Cref{fig:csconstruction,fig:abplacement}.
This construction can be carried out in polynomial time.
\end{construction}

Now we show that \Cref{co:cs} gives us indeed a reduction:

\begin{figure}
\center
\includegraphics[width=0.5\textwidth]{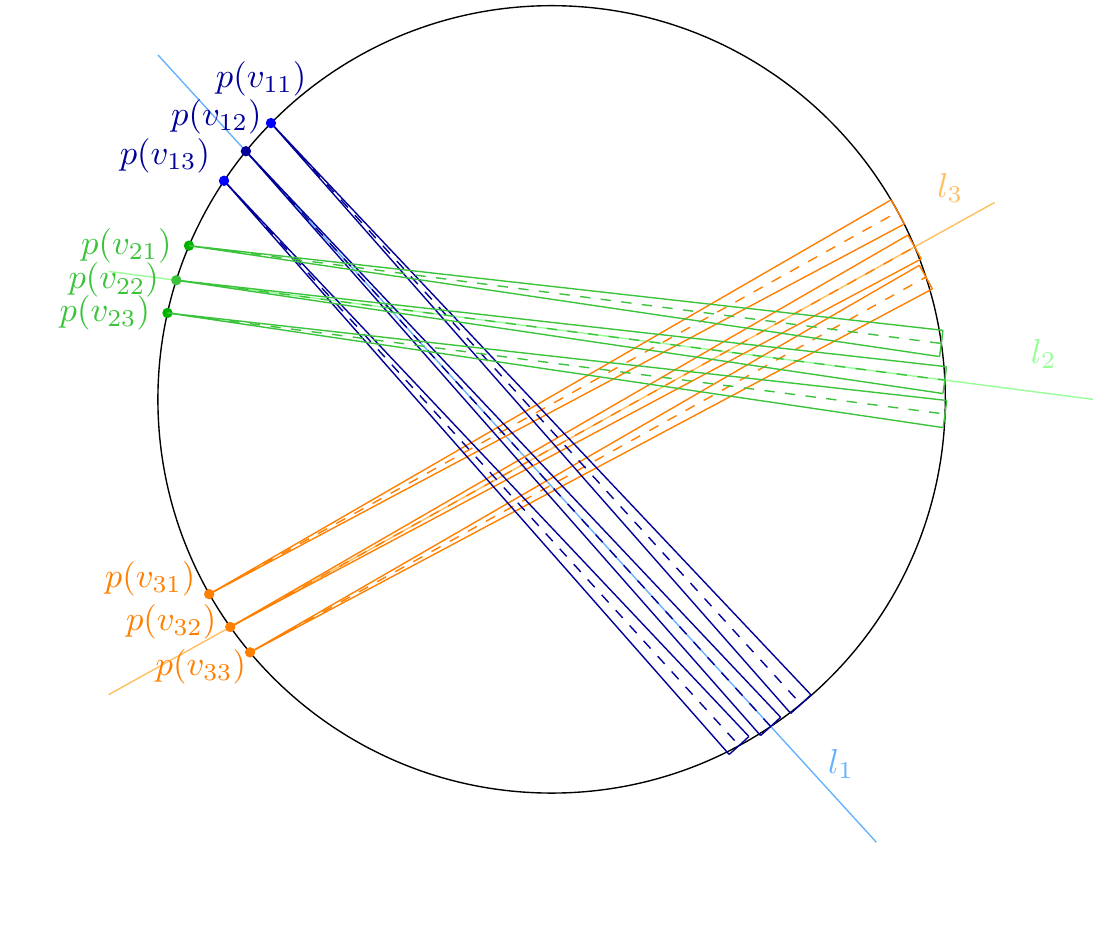}
\caption{Construction of the circular sectors $c_{im}$ 
based on a given line arrangement}
\label{fig:csconstruction}
\end{figure}

\begin{lemma}\label[lemma]{la:cs=>}
Suppose there is a line arrangement 
$\mathcal{L} = \{\ell_1, \dots, \ell_n\}$ 
realizing $\mathcal{D}$, then there is an equiangular, 
wide spread arrangement $\mathcal{C}$ of circular sectors 
realizing $G_L$ as defined in \Cref{co:cs}.
\end{lemma}

\begin{proof}
We construct the containing disk $D$, and the sets of 
intersection points $D_l$ and $D_r$ as in the proof of 
\Cref{thm:1darc}. By $\ell_{im}$, we denote the 
directed line through the bisector of the circular sector $c_{im}$.
Let $\alpha_\text{min}$ be the smallest acute angle 
between any two lines of $\mathcal{L}$.
The angle $\alpha$ for $\mathcal{C}$ will be set 
depending on $\alpha_\text{min}$ and the placement 
of the constructed circular sectors $c_{im}$.

In the first step, we place the circular sectors $c_{i2}$. 
They are constructed such that their apexes are on $q_i^l$ and 
their bisectors are exactly the line segments $\ell_i \cap D$. 
We place $p(c_{i1})$ in clockwise direction next to $p(c_{i2})$ 
onto the boundary of $D$. The distance between $p(c_{i1})$ 
and $p(c_{i2})$ on $\partial D$ is some small $\tau > 0$. 
The point $p(c_{i3})$ is placed in the same way, but in 
counter-clockwise direction from $p(c_{i2})$. The 
bisectors of all $c_{im}$ are parallel.  The radii for $c_{i1}$ and 
$c_{i3}$ are chosen to be the length of the line segments 
$\ell_{i1} \cap D$ and $\ell_{i3} \cap D$.

The distance $\tau$ must be small enough so that no 
intersection of any two original lines lies 
between $\ell_{i1}$ and $\ell_{i3}$. Let $\beta$ be the 
largest angle such that if the angle of all $c_{im}$ is set to 
$\beta$, there is always at least one point in $c_{im}$ 
between the bounding boxes $B$ of two circular sectors with 
consecutively intersecting bisectors. Since $\mathcal{L}$ 
is a simple line arrangement, this is always possible.
The angle $\alpha$ for the construction is now set to 
$\min\left\{\alpha_\text{min}/2, \beta\right\}$. 
This first part of the construction is illustrated 
in \Cref{fig:csconstruction}.

Now we place the remaining circular sectors. Their placement 
can be seen in \Cref{fig:abplacement}. The points $p(a_{km'}^{im})$ 
all lie on $\ell_{im}$ with a distance of $\delta$ to the left of 
the intersection of $\ell_{im}$ and $\ell_{km'}$. By 
``to the left'', we mean that the point lies closer to 
$p(c_{im})$ on the line $\ell_{im}$ than the intersection point.
The distance $\delta$ is chosen small enough such that 
$p(a_{km'}^{im})$ lies inside of all $a_{\mathbf{im}}^{\mathbf{km'}}$ 
that have a larger distance to $p(c_{\mathbf{km'}})$ 
than $p(a_{km'}^{im})$. 
The direction of the circular sector $a_{km'}^{im}$ is set to 
$-u(c_{im})$, and its radius is set to 
$r(a_{o_km'}^{im}) = \dist(p(a_{km'}^{im}), p(c_{im})) + \eps$, 
for $\eps > 0$. This lets $p(c_{im})$ lie on the bisecting line 
segment of every circular sector $a_{km'}^{im}$.
The directions and radii for the $b_{km'}^{im}$ are chosen in the 
same way as for the $a_{km'}^{im}$. The apexes of $b_{km'}^{im}$ 
are placed such that they lie between the corresponding bounding 
boxes $B(c_{km'})$. 
For $\alpha$ small enough, this is always possible.

\begin{figure}
\center
\includegraphics[width=0.68\textwidth]{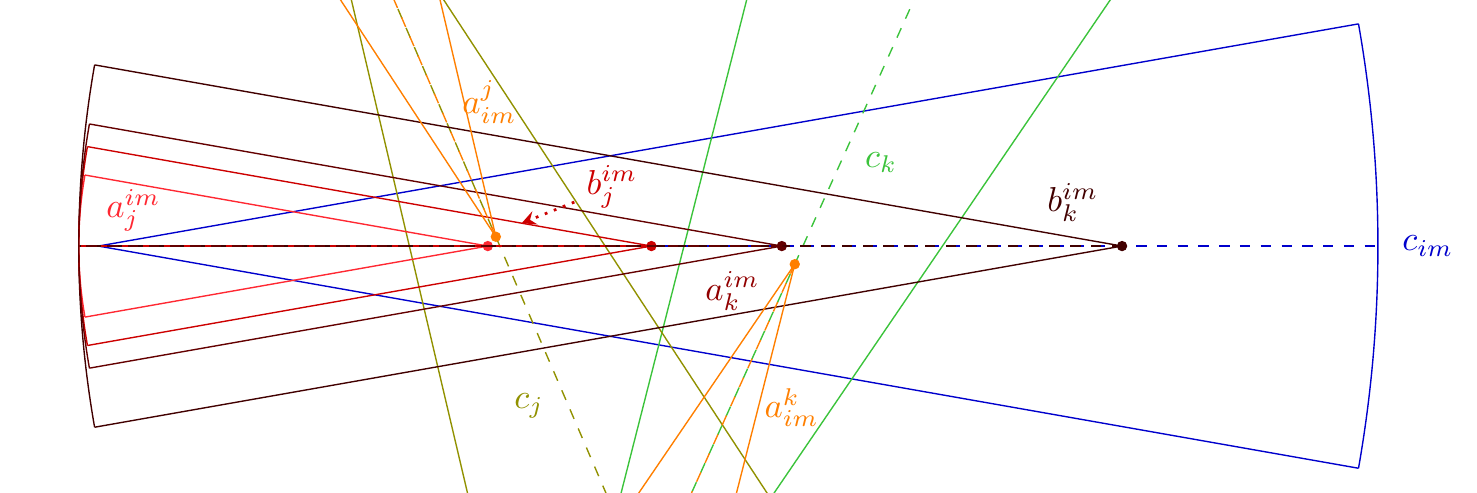}
\caption{Detailed construction inside of one circular sector $c_{im}$.}
\label{fig:abplacement}
\end{figure}

It follows directly from the construction that the generalized 
transmission graph of this arrangement is $G_L$. 
A detailed argument can be found in \Cref{ap:=>}.
\end{proof}

\begin{lemma}\label[lemma]{la:cs<=}
Suppose there is an equiangular, wide spread arrangement $\mathcal{C}$ 
of circular sectors realizing $G_L$ as defined in \cref{co:cs}, then 
there is an arrangement of lines realizing $\mathcal{D}$.
\end{lemma}

\begin{proof}
From $\mathcal{C}$, we  construct an 
arrangement $\mathcal{L} = (\ell_1, \dots, \ell_n)$ 
of lines such that $\mathcal{D}(\mathcal{L}) = \mathcal{D}$ 
by setting $\ell_i$ to the line spanned by $u(c_{i2})$. Now, 
we show that this line arrangement indeed satisfies the description,
e.g., that the intersection order of the lines is as indicated by 
the description. 

All $a_{km'}^{im}$ and $b_{km'}^{im}$ form mutual couples with 
$c_{im}$. Thus, \Cref{la:orderingGadget} can be applied to them. 
It follows that the order of the projections of the apexes of 
the circular sectors is known. In particular, the order of 
projections of the $p(a_{j2}^{i2})$ onto $\ell_i$ is the order 
given by $\mathcal{D}$ and $p(b_{o_j2}^{i2})$ is projected 
between $p(a_{o_j2}^{i2})$ and $p(a_{o_{j+1}2}^{i2})$.

Now, we have to show that the order of intersections of the 
lines corresponds to the order of the projections of the 
$p(a_{j2}^{i2})$. This will be done through a contradiction.
We consider two circular sectors $c_{j2}$ and $c_{k2}$. 
Assume that the order of the projection of the apexes of 
$a_{j2}^{i2}$ and $a_{k2}^{i2}$ onto $\ell_i$ is  $p(a_{j2}^{i2})$, 
$p(a_{k2}^{i2})$, while the order of intersection of the lines is 
$\ell_k$, $\ell_j$. 

Note that by the definition of the edges of $G_L$, $c_{j2}$ and 
$c_{k2}$ share the apexes of $a_{j2}^{k2}$ and $a_{k2}^{j2}$, 
but there is no circular sector they both form a mutual couple with and thus 
the angle between their bisecting line segments is large.

There are two main cases to consider, based on the position 
of the intersection point $p$ of $\ell_{j}$ and $\ell_k$ 
relative to $c_{i2}$:

\paragraph*{Case one $p\notin c_{i2}$:}\label{subsec:caseone}
If $p$ does not lie in $c_{i2}$, then $\ell_j$ and 
$\ell_k$ divide $c_{i2}$ into three parts. 
Let $s_j, s_k$ be the outer line segments of $c_{j2}$ and 
$c_{k2}$ that lie in the middle part of this decomposition.
A schematic of this situation can be seen in \Cref{fig:nonintersect}.

\begin{figure}
\center
\subfloat[Case one, $a_{k2}^{i2}$ cannot reach $a_{j2}^{i2}$.]{
\includegraphics[width=0.48\textwidth]{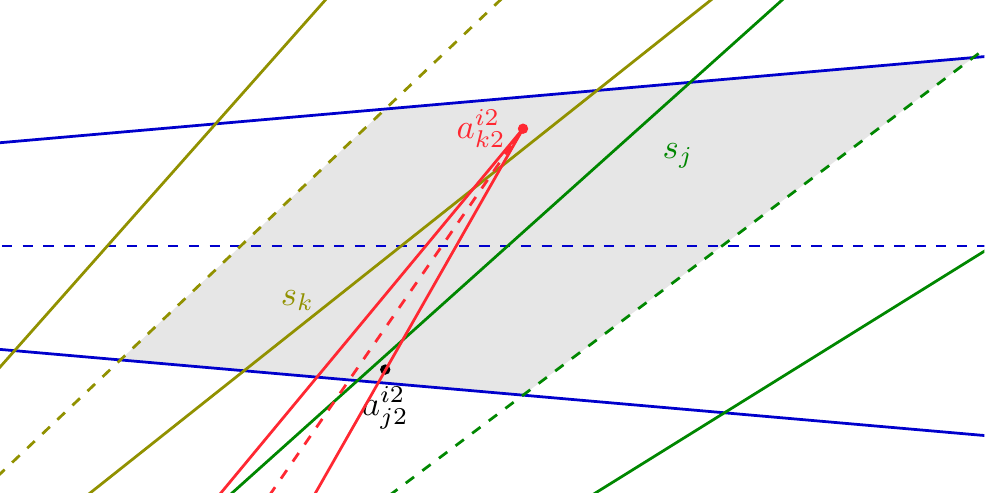}

\label{fig:nonintersect}
}
\subfloat[Case two, $b_{j2}^{i2}$ cannot lie in $F_1$ or $F_3$.]{
\includegraphics[width=0.48\textwidth]{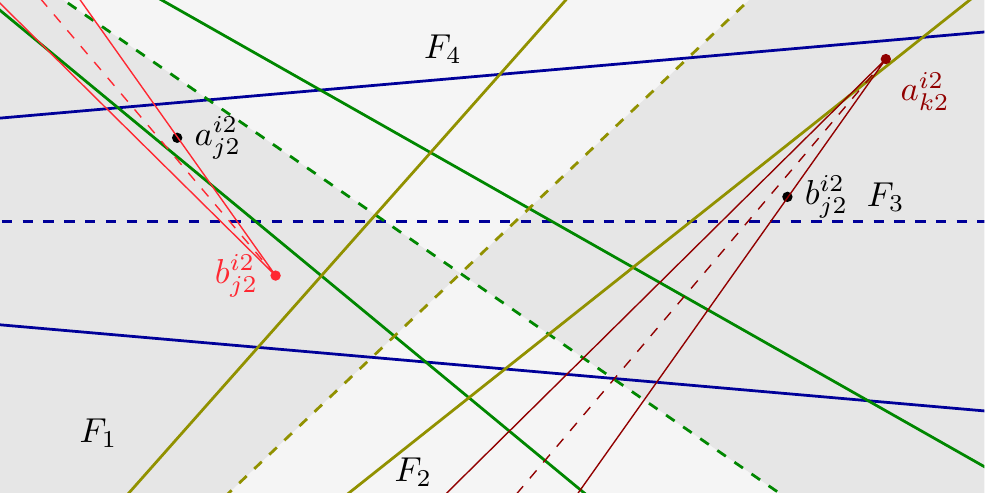}

\label{fig:nof1f3}
}
\end{figure}

From \Cref{ob:outerangle} and since $\mathcal{C}$ is an 
equiangular, wide spread arrangement it follows that 
\(
|\pi - \gamma(s_j, u(c_i))| > 3\alpha/2 \text{ and }
|\pi-\gamma(s_k, u(c_i))| >3\alpha/2
\).

In order to have an intersection order that differs from the 
projection order, the circular sector $a_{k2}^{i2}$ 
has to reach $p(a_{j2}^{i2})$. The latter point is projected 
to the left of $a_{k2}^{i2}$ but lies right of $s_k$. 
The directed line segment $d$ from $p(a_{k2}^{i2})$ to $p(a_{k2}^{j2})$ has to intersect $s_j$ 
and $s_k$, and thus it has to hold that 
$|\pi -\gamma(d, u(c_{i2}))| \geq 3\alpha/2$. The line segment $d$
has to lie inside of $a_{k2}^{i2}$, which is only possible if 
$|\pi-\gamma(u(a_{k2}^{i2}), u(c_i))| > \alpha$. However, this is 
a contradiction to $|\pi-\gamma(u(a_k^i), u(c_i))| \leq \alpha$, 
which follows from \Cref{ob:pairsangles}.

\paragraph*{Case two $p\in c_{i2}$:}
\Wlog, let $u(c_{i2})  = \lambda \cdot (1,0)$, $\lambda>0$, 
and let $\mathcal{F}=\{F_1, F_2, F_3, F_4\}$ be the decomposition 
of the plane into faces induced by $\ell_j$ and $\ell_k$. 
Here, $F_1$ is the face with $p(c_{i2})$, and the faces are 
numbered in counter-clockwise order.

We consider the possible placements of $p(b_{j2}^{i2})$ 
in one of the face. First, we show that $p(b_{j2}^{i2})$ cannot 
lie in $F_1$ or in $F_3$. 
From the form of $E_\text{GO}$, we know that $p(a_{j2}^{i2})$ 
has to be projected left of $p(b_{j2}^{i2})$ and $p(a_{j2}^{i2})$ 
has to lie inside of $b_{j2}^{i2}$; see \Cref{fig:nof1f3} for 
a schematic of the situation.
If $p(b_{j2}^{i2})$ lies in $F_1$, the line segment in $b_{j2}^{i2}$ 
that connects $p(b_{j2}^{i2})$ and $p(a_{j2}^{i2})$ has to cross 
an outer line segment of $c_{j2}$. This yields the same contradiction 
as in the first case.
If $p(b_{j2}^{i2})$ were in $F_3$, an analogous argument 
holds for $p(b_{j2}^{i2})$, which has to lie inside of $a_{k2}^{i2}$.

This leaves $F_2$ and $F_4$ as possible positions for $b_{j2}^{i2}$. 
\Wlog, let $b_{j2}^{i2}$ be located in $F_4$. We divide $c_{j2}$ 
and $c_{k2}$ by $\ell_k$ or $\ell_j$, respectively, 
into two parts, and denote the parts containing the line 
segments that are incident to $F_4$ by $J$ and $K$.
Then, again by using that the arrangement is wide spread, it can 
be seen that $p(a_{j2}^{i2})$ and $p(a_{k2}^{i2})$ are located in $J$ 
and $K$. The possible placement is visualized in \Cref{fig:f4nodiagonal}.

 \begin{figure}
 \subfloat[The localization of $a_{j2}^{i2}$ and $a_{k2}^{i2}$.]{
\includegraphics[width=0.49\textwidth]{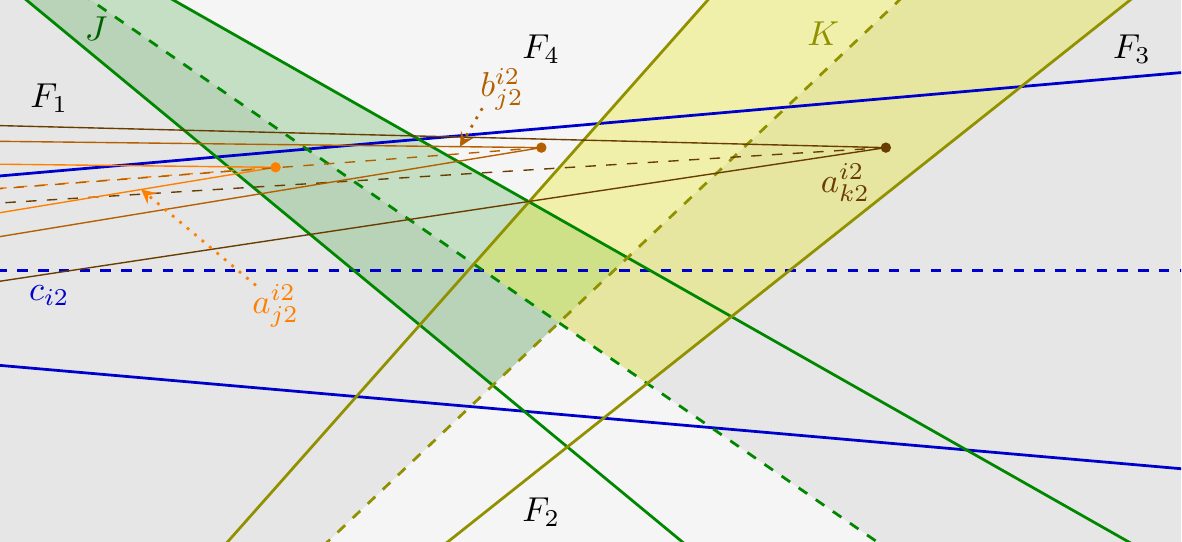}
 \label{fig:f4nodiagonal}
}
  \subfloat[$p$ can not lie in $c_{i1}$.]{
\includegraphics[width=0.49\textwidth]{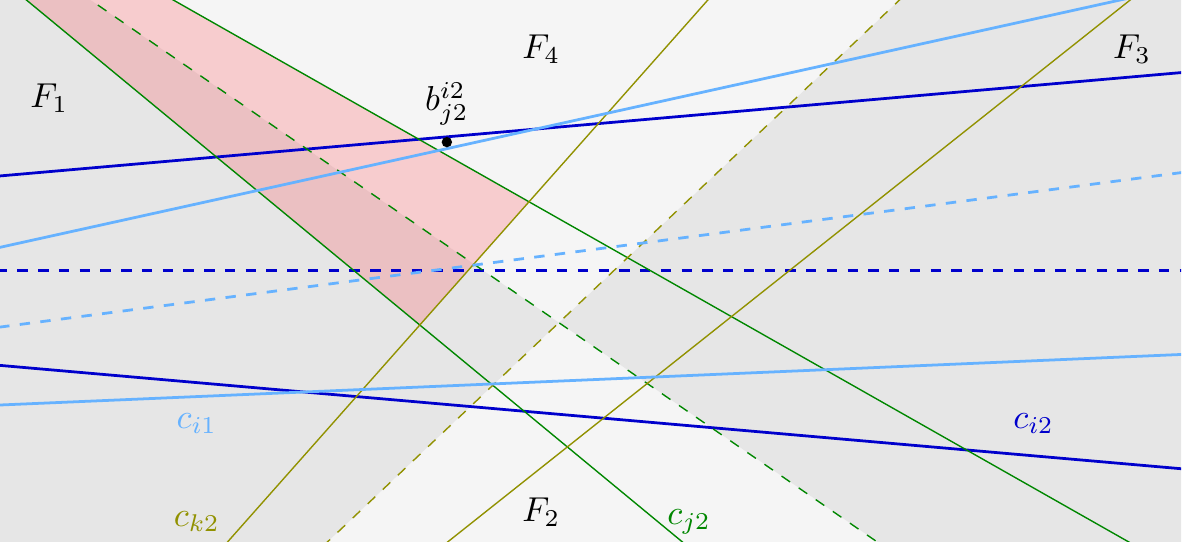}
 \label{fig:f4overlap}
}

 \end{figure}

The argument so far yields that if $p\in c_{i2}$, then the 
intersection order of $\ell_{j}$ and $\ell_k$ with $\ell_i$ 
is the same as the order of projection if $\ell_i$ lies above $p$,
and is the inverse order if $\ell_i$ lies below $p$. 
The uncertainty of this situation is not desirable.
By considering the circular sectors $c_{i1}$ and $c_{i3}$, 
we will now show that such a situation cannot occur.

First, we show that $c_{i1}$ and $c_{i3}$ cannot contain the 
intersection point of $\ell_j$ and $\ell_k$.
\Wlog, assume that the intersection point lies in $c_{i1}$. 
Then, $b_{j2}^{i1}$ is included in either $F_2$ or $F_4$. 
Consider the case that $b_{j2}^{i1}$ lies in $F_4$. 
Since $u(c_{i2}) = \lambda \cdot (1,0)$ and since 
one of the outer line segments of $c_{i2}$ has to lie beneath $p$, 
there is only one outer line segment of $c_{i2}$ that intersects 
$F_4 \setminus (J\cup K)$, $J$ and $K$. There are at most two 
intersection points of this outer line segment with 
$\partial c_{i1}$.  This implies that there is no intersection 
point of $\partial c_{i2}$ and $\partial c_{i1}$ in at least 
one of $J$, $K$, and $F_4 \setminus (J\cup K)$. If there is 
no intersection point, then $c_{i1}$ and $c_{i2}$ overlap in 
this interval. \Wlog, let this area be $J$, and let $c_{i1} \cap J$ 
be fully contained in $c_{i2} \cap J$. Then, $p(a_{j1}^{im})$ 
cannot be placed. Consequently, this situation is not possible. 
The argument is depicted in \Cref{fig:f4overlap}.

If $p(b_{j2}^{i1})$ was included in $F_2$, 
then the order of projection of $p(a_{k2}^{i2})$ and 
$p(a_{j2}^{i2})$ would be the same order as the order of 
intersections of $\ell_j$ and $\ell_k$ with a parallel line 
to $\ell_i$ that lies below $\ell_i$. This order is the 
inverse order of the order of projection in $c_{i2}$. 
Since the order of the projection as defined by $E_\text{GO}$ 
depends only on $k$ and $i$, the order of projection of 
$p(a_{j2}^{im})$ and $p(a_{k2}^{im})$ has to be the same in 
all $c_{im}$. This implies that $p(b_{k2}^{i1})$ is 
not included in $F_2$.

Now, we know that $c_{i1}$ and $c_{i3}$ do not contain 
the intersection point. This implies that the argument 
from the case $p \notin c_{i2}$ can be applied to them 
and the order of intersection in $c_{i1}$ and $c_{i3}$ 
is the same as the order of the projections of 
$p(a_{j2}^{i1})$ and $p(a_{k2}^{i1})$. This order is the 
same in all three $c_{im}$, and thus the bisectors of $c_{i1}$ 
and $c_{i3}$ have to lie on the same side of the intersection 
point. Furthermore, the points $p(a_{j2}^{i1})$ and 
$p(a_{j2}^{i3})$ have to lie in $J$ but outside of $c_{i2}$. 
This implies that $\ell_{i1}$ and $\ell_{i3}$ both intersect 
$\ell_j$ and $\ell_k$ either before $\ell_i$ or after $\ell_i$, 
while $p(b_{j2}^{i2})$ lies in $F_4$.

The edges for the local order define that the order of 
projection onto $\ell_j$ is $p(a_{j2}^{i1})$, $p(a_{j2}^{i2})$,
$p(a_{j2}^{i3})$ (or the reverse), and the analogous statement 
holds for $\ell_k$. This order is not possible with $c_{i1}$ 
and $c_{i3}$, both lying above or below $c_{i2}$, which 
implies that the intersection point cannot lie in $c_{i2}$. 
Since the order of intersection is the same as the order of the 
projection, if $p \notin c_{i2}$ and a situation with $p\in c_{i2}$ 
is not possible, we have shown that 
$\mathcal{D}(\mathcal{L}) = \mathcal{D}$.
\end{proof}

With the tools from above, we can now give the proof of the main result 
of this section:

\begin{theorem}
\SECTOR is hard for \ER.
\end{theorem}

\begin{proof}
The theorem follows from \Cref{co:cs,la:cs=>,la:cs<=}.
\end{proof}

\section{Conclusion}
We have defined the new graph class of 
generalized transmission graphs as a model for 
directed antennas with arbitrary shapes. We showed 
that the recognition of generalized transmission graphs of line segments and a special form of
circular sectors is \ER-hard. 

For the case of circular sectors, we needed to impose
certain conditions on the underlying arrangements.
The wide spread condition in particular 
seems to be rather restrictive. We assume that this condition 
can be weakened, if not dropped, while the problem remains \ER-hard. 

Ours are the first \ER-hardness results on 
directed graphs that we are aware of. We believe that this
work could serve as a starting point 
for a broader investigation into the recognition problem
for geometrically defined directed graph models, and to
understand further what makes these problems hard.

\textbf{Acknowledgments.}
We would like to thank an anonymous reviewer for pointing out a mistake 
in \Cref{ob:pairsangles}.

\bibliographystyle{abbrv}
\bibliography{RecognizingTransmissiongraphs}

\newpage
\appendix

\section{Missing proofs and constructions}

\subsection{Full construction for 
\texorpdfstring{SECTOR}{\SECTOR}}\label{ap:edgedef}

Let the vertices of the construction be defined as in \Cref{co:cs}.
We divide the edges of the graph into \emph{categories}. 
The first category $E_I$ contains the edges that enforce an 
intersection of two circular sectors $c_{im}$ and $c_{km'}$ for $k < l$.
\[
E_I= \Big\{\big(c_{im}, a_{km'}^{im}\big), \big(c_{im}, a^{km'}_{im}\big)
\,\Big\vert\, i\neq k\Big\}.
    \]
The edges $E_C$ enforce that each $a_{km'}^{im}$ and each 
$b_{km'}^{im}$ forms a mutual couple with $c_{im}$.
\[
	E_C= \Big\{\big(a_{km'}^{im}, c_{im}\big), 
    \big(c_{im}, b_{km'}^{im}\big), 
    \big(b_{km'}^{im}, c_{im}\big) \,\Big\vert\, i \neq k\Big\}.
  \]
The edges of $E_\text{GO}$ will enforce the order of the  
projection of the apexes of $a_{o_km'}^{im}$, $a_{o_lm''}^{im}$,
$b_{o_km'}^{im}$, and $b_{o_lm''}^{im}$ for $k > l$ onto the 
bisector of $c_{im}$. They are chosen such that $p(a_{o_km'}^{im})$ 
will be projected closer to $p(c_{im})$ than $p(a_{o_lm''}^{im})$, 
for $k < l$.  Also included in $E_\text{GO}$ are edges that enforce 
that all $p(a_{\mathbf{im}}^{\mathbf{o_km'}})$ are included 
in the circular sectors $a_{o_lm''}^{im}$ and $b_{o_lm''}^{im}$.
\begin{alignat*}{2}
E_\text{GO} &=\phantom{\cup}
\Big\{(a_{o_km'}^{im}, a_{o_lm''}^{im}), (a_{o_km'}^{im},
 a^{\mathbf{o_lm''}}_{\mathbf{im}}),
(a_{o_km'}^{im}, b_{o_lm''}^{im}),&&\\
&\phantom{=\cup
\Big\{}
(b_{o_km'}^{im}, a_{o_lm''}^{im}),
(b_{o_km'}^{im}, a^{\mathbf{o_lm''}}_{\mathbf{im}}) 
&&\Big\vert\, i\neq k,
k>l\Big\}.
\end{alignat*}
The last two categories of edges will enforce the projection order 
of the apexes of $a_{o_k1}^{im}$, $a_{o_k2}^{im}$, $a_{o_k3}^{im}$, 
and $b_{o_k1}^{im}$, $b_{o_k2}^{im}$, $b_{o_k3}^{im}$ onto the 
bisector of $c_{im}$. This order is $a_{o_k1}^{im}$, $b_{o_k1}^{im}$,
$a_{o_k2}^{im}$, $b_{o_k2}^{im}$, $a_{o_k3}^{im}$ $b_{o_k3}^{im}$, 
if $o_k > i$, and the inverse order, otherwise. The edges for 
the first case are $E_\text{LOI}$, and the edges for the second case 
are $E_\text{LOD}$. We set
	\begin{alignat*}{2}
	E_\text{LOI} &=
    \phantom{\cup}\Big\{(a_{o_km'}^{im}, a_{o_km''}^{im}), 
	(a_{o_km'}^{im}, a^{\mathbf{o_km''}}_{\mathbf{im}}),  \\
	&\phantom{=\cup\Big\{}(a_{o_km'}^{im}, b_{o_km''}^{im}), 
	(b_{o_km'}^{im}, b_{o_km''}^{im})  &&\Big\vert\, i\neq k, m''<m', o_k>i\Big\}\\
	&\phantom{=}\cup\Big\{(b_{o_km'}^{im}, a_{o_km''}^{im}), 
	(b_{o_km'}^{im}, a^{\mathbf{o_km''}}_{\mathbf{im}}) &&\Big\vert\, i\neq k,
    m''\leq m', o_k>i\Big\}
	\intertext{and}
	E_\text{LOD} &=
	\phantom{\cup} \Big\{(a_{o_km'}^{im}, a_{o_km''}^{im}),  
	(a_{o_km'}^{im}, a^{\mathbf{o_km''}}_{\mathbf{im}}),\\
		&\phantom{=\cup\Big\{} (a_{o_km'}^{im}, b_{o_km''}^{im}), 
	(b_{o_km'}^{im}, b_{o_km''}^{im})  &&\Big\vert\, i\neq k, m''>m', o_k<i\Big\}\\
	&\phantom{=}\cup \Big\{(b_{o_km'}^{im}, a_{o_km''}^{im}),
	(b_{o_km'}^{im}, a^{\mathbf{o_km''}}_{\mathbf{im}}) &&\Big \vert\, i\neq k,
    m''\geq m', o_k<i\Big\}.
\end{alignat*}
The set of all edges is defined as 
\[
E_L = E_I\cup E_C\cup E_\text{GO} \cup E_\text{LOI} \cup E_\text{LOD}.
\]

\subsection{Remaining proof for 
\texorpdfstring{\Cref{la:cs=>}}{Lemma 8}}\label{ap:=>}
\begin{lemma}
The generalized transmission graph of the arrangement 
$\mathcal{C}$ of circular sectors constructed in \Cref{la:cs=>} is $G_L$
\end{lemma}

\begin{proof}
As $\delta$ is chosen small enough that $a_{km'}^{im}$ and 
$a_{\mathbf{im}}^{\mathbf{km'}}$ lie in $c_{im}$, the edges of 
$E_I$ are created. Since $b_{km'}^{im}$ and $a_{km'}^{im}$ have 
the inverse direction of $c_{im}$ and the radii are large enough, 
$p(c_{im})$ is included in $a_{km'}^{im}$ and in $b_{km'}^{im}$. Hence 
all edges in $E_C$ are created. 

By the choice of the radii and the direction, $a_{o_km'}^{im}$ 
includes all apexes of circular sectors that lie on $\ell_{im}$ 
and closer to $p(c_m)$ than $p(a_{o_km'}^{im})$.  
Furthermore, $\delta$ is small enough such that all 
$a_{\mathbf{im}}^{\mathbf{o_lm}''}$, $l < k$, are 
included in $a_{o_km'}^{im}$. This implies that 
edges from $E_\text{GO}$ are present in the generalized 
transmission graph of $\mathcal{C}$.

The only edges that have not been considered yet are the 
edges in $E_\text{LOI}$ and $E_\text{LOD}$. 
For a circular sector $a_{o_km'}^{im}$ with $o_k > i$, 
the slope of $\ell_{o_k}$ is larger than the slope of $\ell_i$.  
By the counter-clockwise construction, $\ell_{o_k1}$ lies above 
$\ell_{o_k2}$. This implies that the intersection point of 
$\ell_{o_k1}$ and $\ell_{im}$ lies closer to $p(c_{im})$ than 
the intersection points with $\ell_{o_k2}$ or $\ell_{o_k3}$. 
The presence of the edges can now be seen by the same argument 
as for the edges of $E_\text{GO}$.  Symmetrical considerations 
can be made for the edges of $E_\text{LOD}$. 

It remains to show that no additional edges are created. Note that 
all apexes of the circular sectors lie inside of $D$ and that all 
$a_{km'}^{im}\cap D$ and $b_{km'}^{im} \cap D$ are included in the 
boxes $B(c_{im})$.

Since only the apexes of $a_{km'}^{im}$, 
$a_{\mathbf{im}}^{\mathbf{km'}}$, and $b_{km'}^{im}$ lie in $c_{im}$, 
there are no additional edges starting at $c_{im}$.
The rectangles $B(c_{im})$ are disjoint on the boundary of $D$ and all 
$a_{km'}^{im} \cap D$ and $b_{km'}^{im}\cap D$ lie inside of $B(c_{im})$. 
This implies that there are no additional edges ending at $c_{im}$.
Now, we have to consider additional edges starting at $a_{km'}^{im}$ and 
$b_{km'}^{im}$. Note that  $\alpha \leq \pi/4$ enforces that no circular 
sector $a_{km'}^{im}$ or $b_{km'}^{im}$ can reach an apex having a 
larger distance to $p(c_{im})$. Also, note that there are edges 
for all circular sectors with smaller distances in $E_\text{GO}$, 
$E_\text{LOD}$ or $E_\text{LOI}$. This covers all possible 
additional edges. 
\end{proof}
\end{document}